\documentclass[12pt]{article}
\usepackage{amsmath,amssymb,amsfonts,amsthm}

\topmargin- 0.5cm
\newcounter{item}[section]
\newcounter{kirshr}
\newcounter{kirsha}
\newcounter{kirshb}
\newenvironment{enumarab}{\setcounter{kirshb}{1}
\begin{list}{(\arabic{kirshb})}{\usecounter{kirshb}} }{\end{list}}
\newenvironment{mysect}[1]{\vskip8pt\par\noindent\setcounter{item}{1}
\setcounter{equation}{0}{\large\bf\arabic{section}.  #1 }\vskip8pt\nopagebreak\par\nopagebreak }
{\stepcounter{section}\upshape\par}

\newtheorem{theorem}{Theorem}[section]

\newtheorem{corollary}[theorem]{Corollary}
\newtheorem{proposition}[theorem]{Proposition}

\theoremstyle{definition}

\newcommand\overcirc[1]{\raisebox{10pt}{\tiny$\circ$}{\kern-6.5pt}\mbox{$#1$}}
\newcommand\undersym[2]{\raisebox{-6pt}{\tiny$#2$}{\kern-5pt}\mbox{$#1$}}
\begin{document}
\title{\bf Telleparallel Lagrange Geometry and a Unified Field Theory: Linearization of the Field Equations}

\author{M. I. Wanas$^\dagger$, Nabil L. Youssef$^{\,\ddagger}$ and  A. M.  Sid-Ahmed$^{\natural}$\footnote{The authors are members of the Egyptian Relativity Group (ERG): www.erg.net.eg}}
\date{}
\maketitle

\vspace{-1.1cm}
\begin{center}
{$\dagger$ Department of Astronomy, Faculty of Science, Cairo
University\\ CTP of the British University in Egypt (BUE)}
\end{center}
\vspace{-0.9cm}
\begin{center}
 miwanas@sci.cu.edu.eg, \,mwanas@cu.edu.eg
 \end{center}
\vspace{-0.6cm}
\begin{center}
{$\ddagger$ Department of Mathematics, Faculty of Science, Cairo
University}
\end{center}
\vspace{-0.9cm}
\begin{center}
{nlyoussef@sci.cu.edu.eg, \,nlyoussef2003@yahoo.fr}
\end{center}
\vspace{-0.6cm}
\begin{center}
{$\natural$ Department of Mathematics, Faculty of Science, Cairo
University}\end{center}
\vspace{-0.9cm}
\begin{center}
{amrs@mailer.eun.eg, \,amrsidahmed@gmail.com}
\end{center}

\maketitle \vspace{-1cm}

\vspace{1.5cm} \maketitle
\smallskip
\stepcounter{section}

 \noindent{\bf Abstract.}
The present paper is a natural continuation of our previous paper: "Teleparallel Lagrange geometry and a unified field theory, Class. Quantum Grav.,
27 (2010), 045005 (29pp)" \cite{WNA}. In this paper, we apply a linearization scheme on the field equations obtained in \cite{WNA}. Three important results under
the linearization assumption are accomplished. First, the vertical fundamental geometric objects of the
EAP-space loose their dependence on the positional argument
$x$. Secondly, our linearized theory in the Cartan-type case coincides with the GFT in the first order of approximation.
Finally, an approximate solution of the vertical field equations is obtained.




\newpage


\begin{mysect}{Motivation and introduction}

The theory under consideration in the present work is an {\it expansion} of the generalized field theory (GFT) \cite{aaa} (formulated in the
context of Absolute Parallelism (AP-) geometry (cf. \cite{FI}, \cite{b}, \cite{AMR}) to the {\it tangent
bundle  $TM$}. The theory is formulated in the context of Extended Absolute Parallelism (EAP-) geometry \cite{EAP}.
The EAP geometry, combines within its structure, the
geometric richness of the tangent bundle (cf. \cite{GLS}) and the mathematical simplicity of AP-geometry
(cf. \cite{HP}, \cite{b}). The theory, which we refer to as the
Extended Teleparallel Unified Field Theory
(ETUFT), is constructed in a much wider and richer context than the GFT. Accordingly, the suggested
theory has at least the advantages and features of its mother theory (and hopefully more). The GFT has the main
properties required for any physical theory unifying gravity and electromagnetism:

\begin{enumarab}

\item Its material contents is totally induced by geometry \cite{aaa}.

\item All physical (and geometric) quantities of the theory are derived from one entity, namely, the building blocks of the geometry (the fundamental vector
fields forming the parallelization) \cite{aaa}.

\item The theory shows that the charge of gravity (mass/energy) can generate electromagnetism \cite{2010}, \cite{aaaaa}.

\item The theory shows that there is a direct relation between mass and electric charge. In particular, it shows that a part of mass is
electromagnetic in origin and that the mass of the electron is totally electromagnetic in origin \cite{W}.

In addition to the above advantages enjoyed by the GFT, the ETUFT carries within its structure the potentiality of
describing interactions other than gravity and electromagnetism \cite{WNA}. In fact, we
conjecture that the vertical field equations may express some kind of micro-(or quantum) properties.
\end{enumarab}

For the above mentioned reasons and more, we are motivated to study the linearized form of the suggested
field equations in the context of the ETUFT.

\bigskip

In the context of geometric field theories, the linearization scheme, although not covariant, has many advantages.
Among these advantages is the
determination of the constants or/and the parameters characterizing a certain theory.
Another advantage is to test whether a nonlinear theory covers
the domain of a previous, partially successful, linear theory. A third (and important)
one is the attribution of some physical meaning to the geometric objects
used to construct the theory.

\bigskip

The linearization scheme depends mainly on expanding different geometric objects,
used in the construction of the theory, in terms of some small
parameters and then neglecting terms of the second and higher order in these parameters.
The neglect of such small quantities in the
field equations implies  the physical assumption that the field is {weak}. Also, the neglect of similar quantities in the
{equations of motion} reflects the physical assumption that the motion is {slow}. The two assumptions characterize
low energy
systems. This provides us with a tool to test
the theory in low energy regeme. It should be noted that the production of high energies to test some
theories is difficult and
sometimes even impossible because of both technological and budgetary reasons.

\bigskip

In the present work, we are going to linearize the field equation of the ETUFT \cite{WNA}.
In addition to the advantages of the linearization scheme mentioned above, we hope to throw more light
on both the horizontal and vertical geometric objects used in the construction of the ETUFT.
In other words, we hope to illuminate the role of the
extra degrees of freedom implied by the ETUFT. A further advantage, which may be gained from linearization,
is whether {\it one can explore the physical role of the
nonlinear connection characterizing the underlying geometry}. Though
the mathematical role of the nonlinear connection in the derived field equations is
clear\footnote{In fact, the splitting of the field equations into horizontal and vertical counterparts is made
possible due to the existence
of the nonlinear connection.}, the physical aspect of this nonlinear connection needs further
investigation.\footnote{This may be partially achieved if an appropriate
physical interpretation of the directional argument $y$ is given.}

\bigskip

The paper is organized in the following manner.
In section 2,  a short survey of the field equations obtained in the EAP-context is given.
In section 3, we give a brief account on the process of linearization in the the classical AP-context followed by the
process of linearization in the EAP-context. In section 4, we compute the fundamental tensor fields
of the EAP-space under our linearization assumption. In section 5, we study the Cartan-type case under
the linearized condition. In section 6, we give an approximate solution of the vertical field equations,
and finally, we end the
paper by some concluding remarks.

\bigskip

It should be noted that this paper is a (natural) continuation of \cite{WNA}. Accordingly, we will use
the results of \cite{WNA} and stick to its notations.

\end{mysect}


\newpage

\begin{mysect}{A short survey of the field equations in the EAP-context}
We first recall the fundamental tensor fields of the EAP-space. These are given in the following table \cite{WNA}.

\begin{center}{\bf Table 1: Fundamental second rank tensors of EAP-space}\\[0.1 cm]
\footnotesize{\begin{tabular}{|c|c|c|c|}\hline
\multicolumn{2}{|c|}{\hbox{ }} &\multicolumn{2}{c|}{\hbox{ }}\\
\multicolumn{2}{|c|}{{\bf Horizontal}}&\multicolumn{2}{c|}{{\bf Vertical}}
\\[0.4 cm]\hline
&&&\\
{\bf Skew-Symmetric}&{\bf Symmetric}&{\bf Skew-Symmetric}&{\bf Symmetric}
\\[0.4 cm]\hline
&&&\\
$\xi_{\mu\nu}: = \gamma_{\mu\nu}  \!^{\alpha} \!\,_{|\alpha}$&$ \ $&
${\xi}_{ab}: = \gamma_{ab}  \!^{d} \!\,_{||d}$&$ \ $
\\[0.4 cm]\hline
&&&\\
$\gamma_{\mu\nu}: = C_{\alpha}\gamma_{\mu\nu} \!^{\alpha}$& $ \
$&${\gamma}_{ab} : = C_{d}\gamma_{ab} \!^{d}$& $ \ $
\\[0.4 cm]\hline
&&&\\
$\eta_{\mu\nu} := C_{\beta}\,\Lambda^{\beta}_{\mu\nu}$&
$\phi_{\mu\nu} := C_{\beta}\,\Omega^{\beta}_{\mu\nu}$ &$
{\eta}_{ab} := C_{d}\,T^{d}_{ab}$& $
{\phi}_{ab} := C_{d}\,\Omega^{d}_{ab}$
\\[0.4 cm]\hline
&&&\\
$\chi_{\mu\nu} := \Lambda^{\alpha}_{\mu\nu|\alpha}$& $\psi_{\mu\nu}
:= \Omega^{\beta}_{\mu\nu|\beta}$& $ {\chi}_{ab} :=
T^{d}_{ab||d}$& $ {\psi}_{ab} :=
\Omega^{d}_{ab||d}$
\\[0.4 cm]\hline
&&&\\
$\epsilon_{\mu\nu} := C_{\mu|\nu} -
C_{\nu|\mu}$& $\theta_{\mu\nu} :=  C_{\mu|\nu}
+ C_{\nu|\mu}$& ${\epsilon}_{ab} :=
C_{a||b} - C_{b||a}$& ${\theta}_{ab} :=
C_{a||b} + C_{b||a}$
\\[0.4 cm]\hline
&&&\\
\tiny{$U_{\mu\nu} :=
\gamma^{\beta}_{\alpha\mu}\gamma^{\alpha}_{\nu\beta}
 - \gamma^{\beta}_{\mu\alpha}\gamma^{\alpha}_{\beta\nu}$}&
\tiny{$h_{\mu\nu}: =
\gamma^{\beta}_{\alpha\mu}\gamma^{\alpha}_{\nu\beta} +
\gamma^{\beta}_{\mu\alpha}\gamma^{\alpha}_{\beta\nu}$}&

\tiny{${U}_{ab}: =
\gamma^{c}_{da}\gamma^{d}_{bc} -
\gamma^{c}_{ad}\gamma^{d}_{cb}$}&
\tiny{${h}_{ab} :=
\gamma^{c}_{da}\gamma^{d}_{bc} +
\gamma^{c}_{ad}\gamma^{d}_{cb}$}
\\[0.4 cm]\hline
&&&\\
$ \ $&$\sigma_{\mu\nu} :=
\gamma^{\beta}_{\alpha\mu}\gamma^{\alpha}_{\beta\nu}$& $ \ $&
${\sigma}_{ab}: =
\gamma^{c}_{da}\gamma^{d}_{cb}$

\\[0.4 cm]\hline
&&&\\
$ \ $&$\omega_{\mu\nu} :=
\gamma^{\beta}_{\mu\alpha}\gamma^{\alpha}_{\nu\beta}$& $  \ $&
${\omega}_{ab} :=
\gamma^{c}_{ad}\gamma^{d}_{bc}$
\\[0.4 cm]\hline
&&&\\
$ \ $&$\alpha_{\mu\nu} := C_{\mu}C_{\nu}$& $ \ $& $
{\alpha}_{ab} := C_{a}C_{b}$
\\[0.4 cm]\hline
\end{tabular}}
\end{center}

\bigskip

We now give a short survey of the field equations obtained in \cite{WNA}. We take for the horizontal field equations a Lagrangian similar in form (but not in content) to that used by Mikhail and
Wanas in their construction of the GFT \cite{aaa}. We also assume that the nonlinear connection is {\it independent of the horizontal counterparts of the
fundamential vector fields forming the parallelization.}\footnote{This condition is actually satisfied under the Cartan-type condition \cite{WNA}.}

\bigskip

In view of the above, for {\bf the horizontal field equations}, we start with the following scalar Lagrangian: Let
$${\cal H} = |\lambda| g^{\mu\nu} H_{\mu\nu},$$ where
\begin{equation}H_{\mu\nu} := \Lambda^{\alpha}_{\epsilon\mu}\Lambda^{\epsilon}_{\alpha\nu} - C_{\mu}C_{\nu}.\end{equation}

The Euler-Lagrange equations \cite{GLS} for this Lagrangian are given by
\begin{equation}\label{ELEs}\frac{\delta {\cal H}}{\delta \lambda_{\beta}}: = \frac{\partial {\cal H}}{\partial \lambda_{\beta}} -
\frac{\partial}{\partial x^{\gamma}}\bigg(\frac{\partial {\cal H}}
{\partial \lambda_{\beta, \gamma}}\bigg) - \frac{\partial}{\partial y^{a}}\bigg(\frac{\partial {\cal H}}
{\partial \lambda_{\beta; a}}\bigg) = 0.\end{equation}

 Setting
\begin{equation}\label{xcx}E^{\beta}_{\sigma}: = \frac{1}{|\lambda|}\bigg(\frac{\delta {\cal H}}{\delta \,\undersym{\lambda}{j}_{\beta}}\bigg)\,
\undersym{\lambda}{j}_{\sigma},\end{equation}
the Euler-Lagrange equations (\ref{ELEs}) take the form
\begin{equation}\begin{split}\label{EEEE}\\[- 0.5 cm]0&=E^{\beta}_{\sigma} = \delta^{\beta}_{\sigma}H - 2H^{\beta}_{\sigma}
- 2C_{\sigma}C^{\beta} - 2\delta^{\beta}_{\sigma}C^{\epsilon}\!\,_{|\epsilon} +
2\delta^{\beta}_{\sigma}C^{\epsilon}C_{\epsilon} -2C^{\epsilon}\Lambda^{\beta}_{\epsilon\sigma}\\& \ \ \ \ \ \ \ \ \ \ \
\ + \, 2g^{\alpha\beta}C_{\sigma|\alpha} - 2g^{\gamma\alpha}\Lambda^{\beta}_{\sigma\alpha|\gamma}
 - 2N^{a}_{\alpha; a}(\Lambda^{\alpha}\,_{\sigma}\,^{\beta}-  \Lambda^{\beta}\,_{\sigma}\,^{\alpha})\\&
\ \ \ \ \ \ \ \ \ \ \ \ + \, 2C^{\nu}(\delta^{\beta}_{\sigma}N^{a}_{\nu; a} - \delta^{\beta}_{\nu}N^{a}_{\sigma; a}) + 2g^{\alpha\beta}\{\mathfrak{S}_{\sigma, \alpha, \epsilon}C^{\epsilon}_{\alpha a}
R^{a}_{\sigma\epsilon}\}\end{split}\end{equation}
Lowering the index $\beta$ in (\ref{EEEE}) and renaming the indices, we get
\begin{equation}\begin{split}\label{EEEXz}\\[- 0.7 cm]0& = E_{\mu\nu} := g_{\mu\nu}H - 2H_{\mu\nu} -
2C_{\mu}C_{\nu} - 2g_{\mu\nu}(C^{\epsilon}\!\,_{|\epsilon} - C^{\epsilon}C_{\epsilon})- 2C^{\epsilon}\Lambda_{\mu\epsilon\nu}
+ 2C_{\nu|\mu}\\& \ \ \ \ \ \ \ \ \ \ \ \ \ \ - \ 2g^{\epsilon\alpha}\Lambda_{\mu\nu\alpha|\epsilon}
 - 2N^{a}_{\epsilon; a}(\Lambda^{\epsilon}\,_{\nu\mu} -  \Lambda_{\mu\nu}\,^{\epsilon})
+ 2g_{\mu\nu}C^{\epsilon}N^{a}_{\epsilon; a} - 2C_{\mu}N^{a}_{\nu; a}\\& \ \ \ \ \ \ \ \ \ \ \ \ \ \ + 2 \ \mathfrak{S}_{\mu, \nu, \epsilon}C^{\epsilon}_{\mu a}
R^{a}_{\nu\epsilon}.\end{split}\end{equation}
{This is the {generalized horizontal field equations} in the context of the EAP-geometry}.

\bigskip

Considering the {symmetric} part of (\ref{EEEXz}), denoting $N_{\beta} := N^{a}_{\beta; a}$, it is found that
\begin{equation}\label{Tarek}\begin{split}  0 &= E_{(\mu\nu)}: = (g_{\mu\nu} \ \overcirc{\cal R}  - 2 \ \overcirc{R}_{(\mu\nu)})
+ g_{\mu\nu}(\sigma - h - Q) - \ 2(\sigma_{\mu\nu} - h_{\mu\nu} - Q_{(\mu\nu)})\\
& \ \ \ \ \ \ \ \ \ \ \ \ \ \ \ +  \ N^{\beta}
(\Lambda_{\mu\nu\beta} + \Lambda_{\nu\mu\beta}) +  2g_{\mu\nu}C^{\beta}N_{\beta} -
(C_{\mu}N_{\nu} + C_{\nu}N_{\mu}),\end{split}\end{equation}
which represents the {symmetric part of the generalized horizontal field equations.} 
\bigskip

\noindent Setting

\begin{equation}\label{newtensors}M_{\mu\nu} := N^{\beta}\Lambda_{\mu\nu\beta}, \ \ \ \ \ Z_{\mu\nu} := C_{\mu}N_{\nu}, \ \ \ \
Z := g^{\mu\nu}Z_{\mu\nu},\end{equation}
we conclude, by (\ref{Tarek}), that
\begin{equation}\label{first order}\overcirc{R}_{(\mu\nu)} - \frac{1}{2}\, g_{\mu\nu}\,\,\overcirc{\cal R} = T_{(\mu\nu)};\end{equation}
\begin{equation}\label{ems}T_{(\mu\nu)} := \frac{1}{2}\,g_{\mu\nu}(\sigma - h - Q + 2Z)   - (\sigma_{\mu\nu} - h_{\mu\nu} - Q_{(\mu\nu)}) +
(\frac{1}{2}\,N_{\beta}\Omega^{\beta}_{\mu\nu} - Z_{(\mu\nu)}).\end{equation}
$T_{(\mu\nu)}$ is interpreted as the
{generalized energy momentum tensor}, constructed from the
symmetric tensors $\sigma_{\mu\nu}$, $h_{\mu\nu}$, $N_{\beta}\Omega^{\beta}_{\mu\nu}$, $Q_{(\mu\nu)}$ and $Z_{(\mu\nu)}$.

\noindent Consequently, the horizontal Einstein tensor has the form
\begin{equation}\begin{split}\label{cxa} \overcirc{J}_{\mu\nu} & := \ \overcirc{R}_{\mu\nu} - \frac{1}{2}\, g_{\mu\nu} \ \overcirc{\cal R}\\
& \ = \, \{\frac{1}{2}\, g_{\mu\nu}(\sigma - h) + (h_{\mu\nu} - \sigma_{\mu\nu})\} + \frac{1}{2}\,g_{\mu\nu}(2Z - Q) + (\frac{1}{2} \, N_{\beta}\,
\Omega^{\beta}_{\mu\nu} -
Z_{(\mu\nu)} + Q_{(\mu\nu)})\\
& \ \ \, \ \ + \ \frac{1}{2}\, \mathfrak{S}_{\mu, \nu, \alpha} \ \overcirc{C}^{\alpha}_{\mu a} R^{a}_{\nu \alpha},\\[- 0.5 cm]\end{split}\end{equation}
which is subject to the identity 
\begin{equation}\label{HESxx} \,\overcirc{J}^{\mu}\!\,_{\sigma{o\atop|}\mu} = R^{a}_{\sigma\mu}\,\,\overcirc{P}^{\mu}_{a} +
\frac{1}{2}\,R^{a}_{\alpha\mu}\,\,\overcirc{P}^{\alpha\mu}\!\,_{\sigma a}.\end{equation}

On the other hand, considering the {skew-symmetric} part of equation (\ref{EEEXz}), it is found that
\begin{equation}\label{tSM} 0 = E_{[\mu\nu]} = 2\{(\gamma_{\mu\nu} - \epsilon_{\mu\nu} - \xi_{\mu\nu} +
N_{\beta}\Lambda^{\beta}_{\mu\nu}) + (M_{[\mu\nu]} - Z_{[\mu\nu]})\} +  3\mathfrak{S}_{\mu, \nu, \epsilon}C^{\epsilon}_{\nu a}R^{a}_{\epsilon\mu}.\end{equation}

\noindent Let us define
\begin{equation}\begin{split}\label{tEMF}F_{\mu\nu}: &=  (\gamma_{\mu\nu} - \xi_{\mu\nu} + \eta_{\mu\nu} + N_{\beta}\Lambda^{\beta}_{\mu\nu})
+ (M_{[\mu\nu]} - Z_{[\mu\nu]}) + \frac{3}{2}\,\mathfrak{S}_{\mu, \nu, \epsilon}C^{\epsilon}_{\nu a}R^{a}_{\epsilon\mu} \\
& = (\gamma_{\mu\nu} - \xi_{\mu\nu} + \eta_{\mu\nu}) + N^{\beta}(\gamma_{\mu\nu\beta} + \Lambda_{\beta\mu\nu}) +
(\frac{1}{2}\,N^{\beta}\Lambda_{\beta\mu\nu} - Z_{[\mu\nu]})\\& \ \ \
\, + \ \frac{3}{2}\,\mathfrak{S}_{\mu, \nu, \epsilon}C^{\epsilon}_{\nu a}R^{a}_{\epsilon\mu},\end{split}\end{equation}
then we get from (\ref{tSM}) and (\ref{tEMF})
\begin{equation}\label{CUzx}F_{\mu\nu} = \delta_{\nu}C_{\mu} - \delta_{\mu}C_{\nu}\end{equation}
and
\vspace{0.2 cm}\begin{equation}\label{GMESz}\mathfrak{S}_{\mu, \nu, \sigma} \, F_{\mu\nu{o\atop|}\sigma} = - \,\mathfrak{S}_{\mu, \nu, \sigma} \, R^{a}_{\mu\nu}
\dot{\partial_a}C_{\sigma}.\end{equation}
Accordingly, if $F_{\mu\nu}$ is interpreted as the {horizontal electromagnetic field}, then (\ref{GMESz}) represents the
{generalized horizontal Maxwell's equations} and, in view of  (\ref{CUzx}), $C_{\mu}$ is the {horizontal
electromagnetic potential}. Again, by (\ref{tEMF}), $F_{\mu\nu}$ is
constructed from the horizontal skew-symmetric fundamental tensors of the EAP-space (Table 1) together with the
skew-symmetric tensors
$N^{\beta}\gamma_{\mu\nu\beta}$, $N^{\beta}\Lambda_{\beta\mu\nu}$,
$\mathfrak{S}_{\mu, \nu, \epsilon}C^{\epsilon}_{\nu a}R^{a}_{\epsilon\mu}$
and $Z_{[\mu\nu]}$. It is thus constructed from a purely geometric standpoint.

\bigskip

Moreover, if
\begin{equation}J^{\mu} := F^{\mu\nu}\!\,_{{o\atop{|}}\nu},\end{equation}
then it is deduced that
\begin{equation}\label{CONSz}J^{\mu}\!\,_{{o\atop{|}}\mu} = \frac{1}{2}\,\{ F^{\epsilon\mu}(\,\overcirc{R}_{\mu\epsilon} - \,\overcirc{R}_{\epsilon\mu}) +
R^{a}_{\mu\nu}F^{\mu\nu}\!\,_{{o\atop{||}}a}\}.\end{equation}
In the case where the nonlinear connection $N^{\alpha}_{\mu}$ is integrable \cite{GLS}, (\ref{CONSz}) can
be viewed as a generalized conservation law and
$J^{\mu}$ as the {generalized horizontal current density}.

\bigskip

For {\bf the vertical field equations}, we consider a scalar Lagrangian formed of vertical entities, namely,
$${\cal V} := {||\lambda||} g^{ab} V_{ab},$$ where
\begin{equation}V_{ab} := T^{d}_{ea}T^{e}_{db} - C_{a}C_{b}.\end{equation}
The Euler-Lagrange equations reduce to
\begin{equation}\frac{\partial {\cal V}}{\partial \lambda_{b}} - \frac{\partial}{\partial y^{e}}\bigg(\frac{\partial {\cal V}}
{\partial \lambda_{b; \, e}}\bigg) = 0.\end{equation}
In this case, we obtain,
\begin{equation}\begin{split}\label{xax} 0 &= E_{ab}: = g_{ab} V - 2V_{ab} - 2g_{ab} (C^{e}\!\,_{||e} - C^{e}C_{e})
- 2C_{a}C_{b} - 2C^{e} T_{aeb}\\& \ \ \ \ \ \ \ \ \ \ \ \ \ + 2C_{b||a} -
2g^{de}T_{abe||d}.\end{split}\end{equation}
Considering the symmetric part of (\ref{xax}), we get
\begin{equation} 0 = E_{(ab)}: = (g_{ab} \ \overcirc{\cal S}  - 2 \ \overcirc{S}_{ab}) + g_{ab}(\bar{\sigma} - \bar{h}) -
2(\sigma_{ab} - h_{ab}),\end{equation}
so that
\begin{equation}\label{SYMx}\ \overcirc{S}_{ab} - \frac{1}{2} \, h_{ab} \ \overcirc{\cal S} = T_{ab},\end{equation}
\begin{equation}\label{xems}T_{ab}: = \frac{1}{2}\,g_{ab}(\bar{\sigma} - \bar{h})  - (\sigma_{ab} - h_{ab}),\end{equation}
where
\begin{equation}\label{cons}T^{a}\,_{b{o\atop{||}} a} = 0.\end{equation}
Consequently, in view of (\ref{SYMx}) and (\ref{cons}), $T_{ab}$ could be interpreted as the {generalized vertical}
{energy-momentum tensor}, which is, according to (\ref{xems}),
constructed from the vertical
symmetric fundamental tensors of the EAP-space (Table 1).

\bigskip

Considering the skew-symmetric part of (\ref{xax}), we conclude that if
\begin{equation}\label{MAXTx}F_{ab} :=  (\gamma_{ab} - \xi_{ab} + \eta_{ab}),\end{equation}
then
\begin{equation}\label{po}F_{ab} = \dot{\partial_b}C_{a} - \dot{\partial_a}C_{b}.\end{equation}
$F_{ab}$ is interpretted as the {generalized vertical} {electromagnetic field} and
$C_{a}$ as the {generalized vertical} {electromagnetic potential}. Moreover, $F_{ab}$ satisfies
the differential identity

\begin{equation}\label{GMEz}\mathfrak{S}_{a, b, c} \, F_{ab{o\atop||}c} = 0.\end{equation}

It is clear, by (\ref{MAXTx}), that $F_{ab}$ is constructed from the vertical skew-symmetric tensors of Table 1.

\bigskip

Finally, if we set
\begin{equation}\label{u}J^{a} := F^{ab}\!\,_{{o\atop{|}}b}\end{equation}
then,  similar to (\ref{CONSz}), $J^{a}$ represents the {generalized vertical current density} and satisfies the
{conservation law}
\begin{equation}\label{CON}J^{a}\!\,_{{o\atop{|}}a} = 0.\end{equation}
\end{mysect}


\begin{mysect}{Linearization scheme in the EAP-context}
We first give a brief account of the process of linearization in the the classical AP-context \cite{aaaa}.
\,The vector fields $\,\undersym{\lambda}{i}_{\mu}$
in the Minkowski space of special relativity are given by  \begin{equation}\,\undersym{\lambda}{i}_{\mu} = \,\undersym{\delta}{i}_{\mu},\end{equation}
where $i\in\{1, \ldots, 4\}$, $\mu\in\{1, \ldots, 4\}$ and $\,\undersym{\delta}{i}_{\mu}$ is the Kronecker delta.
To get a space which differs slightly from the flat space, it is assumed that
\begin{equation}\label{nwa}\,\undersym{\lambda}{i}_{\mu} = \,\undersym{\delta}{i}_{\mu} + \epsilon\,\undersym{h}{i}_{\mu}\end{equation}
where $\,\,\undersym{h}{i}_{\mu}\in C^{\infty}(M)$\footnote{the algebra of smooth functions on $M$} represents {\bf perturbation} terms and the
parameter $\epsilon$ is assumed to be of small magnitude compared to unity.
In this scheme, each geometric object $G$ defined in the AP-space can be expressed in the form
\begin{equation}G = \sum_{r = 0}^{p} \epsilon^{r} G^{(r)} \ \ \ \ \ (p\in \{1,2,\ldots \}),\end{equation}
where $r$ denotes the power of $\epsilon$ and $G^{(r)}$ is the coefficient of $\epsilon^{r}$. For example, using (\ref{nwa}), the
metric tensor $g_{\mu\nu}$ is found to be
\begin{eqnarray*}g_{\mu\nu} &=&\,\undersym{\lambda}{i}_{\mu}\,\,\undersym{\lambda}{i}_{\nu}\\ &=&
(\,\undersym{\delta}{i}_{\mu} + \epsilon\,\undersym{h}{i}_{\mu})
(\,\undersym{\delta}{i}_{\nu} + \epsilon\,\undersym{h}{i}_{\nu})\\&=&\delta_{\mu\nu} + \epsilon
(\,\undersym{\delta}{i}_{\mu}\,\,\undersym{h}{i}_{\nu} +  \,\undersym{\delta}{i}_{\nu}\,\,\undersym{h}{i}_{\mu}) +  \ \epsilon^{2}(\,\undersym{h}{i}_{\mu}
\,\,\undersym{h}{i}_{\nu})\\
\\[ - 1 cm]\end{eqnarray*}

This {linearization} proceedure was carried out by Mikhial and Wanas for the
GFT (\cite{aaaa}, \cite{aa}). The results obtained
showed a perfect agreement of GFT, in the first order of approximation, with both
general relativity and Maxwell's theories. In addition, the
linearized theory led to the prediction
of new features, namely, the existence of a
{mutual interaction} between both gravitational and electromagnetic fields \cite{aaaa}.

\bigskip

In the present paper, we carry out a linearization of the field equations obtained in \cite{WNA}.
We will use the same conventions usually
followed by physicists in the current literature in which, for example, mesh and world indices are mixed.
The treatment here is, therefore, somewhat
less rigorous than our previous paper \cite{WNA}.

\bigskip

In analogy to the above linearization process,
we assume, in the context of EAP-space, that
\begin{equation}\label{linear}\,\undersym{\lambda}{i}_{\mu} = \,\undersym{\delta}{i}_{\mu} + \epsilon\,\undersym{h}{i}_{\mu}(x) +
\varepsilon\,\undersym{k}{i}_{\mu}(y)
, \ \ \ \ \ \ \,
\undersym{\lambda}{i}_{a} = \,\undersym{\delta}{i}_{a} + \epsilon\,\undersym{h}{i}_{a}(x) + \varepsilon\,\undersym{k}{i}_{a}(y),\end{equation}
where $i\in\{1, \ldots, 4\}$, $\mu\in\{1, \ldots, 4\}$, $a\in \{1, \ldots, 4\}$ and
$\,\,\undersym{h}{i}_{\mu}, \,\,\undersym{k}{i}_{\mu}, \,\,\undersym{h}{i}_{a}, \,\,\undersym{k}{i}_{a}\in C^{\infty}(TM)$ represent
horizontal and vertical perturbation terms. Here $\,\undersym{h}{i}_{\mu}, \,\,\undersym{h}{i}_{a}$ are functions of the
{positional} argument $x$ only whereas $\,\undersym{k}{i}_{\mu}, \,\,\undersym{k}{i}_{a}$ are functions of the {directional}
argument $y$ only.
Moreover, the
parameters $\epsilon$ and $\varepsilon$ are assumed to be of small magnitude compared to unity:
$O(\epsilon)\backsimeq O(\varepsilon) \preccurlyeq 1$ so that all terms of
order $\epsilon^{2}$, $\epsilon\varepsilon$, $\varepsilon^{2}$ or higher
orders can be neglected. This means that we are dealing with
a {\bf weak field}. Finally, $e$ will denote either $\epsilon$ or
$\varepsilon$ so that, for example, $O(e^{2})$
will mean either $O(\epsilon^{2})$, $O(\epsilon\varepsilon)$ or
$O(\varepsilon^{2})$.

\bigskip

We interpret $x^{1}, x^{2}$ and $x^{3}$ as {space} coordinates wheresas $x^{4}$ is
taken to be the {time} coordinate.
On the other hand, the vector $y^{a}$ is attatched as an {\it internal} variable to each point $x^{\mu}$. In this sense,
$y^{a}$ may be regarded as the {spacetime fluctuation} (the micro-internal freedom) associated to
the point $x^{\mu}$ (\cite{GLS}, \cite{V}.) We will return to the interpretation of
the directional argument $y$ later on.
\end{mysect}


\begin{mysect}{First order approximation}

We now carry out the task of computing the fundamental tensors of the EAP-space under the
linearization assumption (\ref{linear}). The vertical (resp. horizontal) counterpart is obtained under
no condition (resp. in the
Cartan type case). Though, we have dealt with many cases in our previous paper \cite{WNA},
it is {\it the Cartan type case that
can lend itself to the process of linearization}. This is because the nonlinear connection and hence
all geometric objects considered are expressed explictly in terms
of the fundamental vector fields. On the other hand, a linearization of the horizontal field equations in
the Berwald type case {gives nothing new}, since the derived horizontal field equations in this case
actually {coincide with those of the GFT} \cite{WNA}.\footnote{This implies, in particular,
that both Maxwell's and general relativity theories
are an outcome of our field equations under the Berwald condition.}

\bigskip

In preparation to what follows, we set
\begin{equation}z_{\mu\nu}: = (\,\undersym{h}{\mu}_{\nu} + \,\undersym{h}{\nu}_{\mu}),  \ \ \
w_{\mu\nu} := (\,\undersym{k}{\mu}_{\nu} + \,\undersym{k}{\nu}_{\mu}); \ \ \ \ \ \undersym{h}{\mu}_{\nu} := \,\undersym{\delta}{i}_{\mu}\,\,
\undersym{h}{i}_{\nu}, \ \ \ \undersym{k}{\mu}_{\nu} := \,\undersym{\delta}{i}_{\mu}\,\,
\undersym{k}{i}_{\nu}\\[- 0.4 cm]\end{equation}

\begin{equation}z_{ab}: = (\,\undersym{h}{a}_{b} + \,\undersym{h}{b}_{a}), \ \ \ w_{ab}: =
(\,\undersym{k}{a}_{b} + \,\undersym{k}{b}_{a}); \ \ \ \ \ \undersym{h}{a}_{b} := \,\undersym{\delta}{i}_{a}\,\,
\undersym{h}{i}_{b}, \ \ \ \undersym{k}{a}_{b} := \,\undersym{\delta}{i}_{a}\,\,
\undersym{k}{i}_{b}.\end{equation}

Then, in view of (\ref{linear}), we obtain

\begin{theorem}\label{ATEF} To the first order of approximation, we have
\begin{description}

\item [(a)]$\,\undersym{\lambda}{i}^{\mu}\backsimeq \,\,\undersym{\delta}{i}^{\mu} -
\epsilon \,\,\undersym{h}{\mu}_{i}(x) - \varepsilon\,\,\undersym{k}{\mu}_{i}(y)$; \ \ \ $\undersym{\lambda}{i}^{a}\backsimeq \,\,
\undersym{\delta}{i}^{a} -
\epsilon \,\,\undersym{h}{a}_{i}(x) - \varepsilon\,\,\undersym{k}{a}_{i}(y).$

\item [(b)] $g_{\mu\nu} \backsimeq \delta_{\mu\nu} + \epsilon z_{\mu\nu}(x) + \varepsilon w_{\mu\nu}(y); \ \ \ \ g^{\mu\nu}\backsimeq \delta_{\mu\nu} -
\epsilon z_{\mu\nu}(x) - \varepsilon w_{\mu\nu}(y).$

\item [(c)] $g_{ab}\backsimeq \delta_{ab} + \epsilon z_{ab}(x) + \varepsilon w_{ab}(y); \ \ \ \ \ \ g^{ab}\backsimeq \delta_{ab} -
\epsilon z_{ab}(x) - \varepsilon w_{ab}(y).$

\item [(d)]$C^{a}_{bc} \backsimeq \varepsilon(\,\undersym{k}{a}_{b; c})(y).$

\item [(e)] $T^{a}_{bc} \backsimeq \varepsilon(\,\undersym{k}{a}_{b; c} - \,\undersym{k}{a}_{c; b})(y).$

\item [(f)]$C_{b} \backsimeq \varepsilon (\,\undersym{k}{a}_{b; a} - \,\undersym{k}{a}_{a; b})(y).$

\item [(g)] $\,\overcirc{C}^{a}_{bc} \backsimeq \frac{1}{2}\,\varepsilon(w_{ab; c} + w_{ac; b} - w_{bc; a})(y).$

\item [(h)]$\gamma^{a}_{bc} \backsimeq \varepsilon\{\,\undersym{k}{a}_{b; c} - \frac{1}{2}\,(w_{ab; c} +
w_{ac; b} - w_{bc; a})\}(y).$

\item [(j)]$\Omega^{a}_{bc} \backsimeq \varepsilon
\{(\,\undersym{k}{a}_{b;c} + \,\undersym{k}{a}_{c; b}) - (w_{ab; c} +
w_{ac; b} - w_{bc; a})\}(y).$

\end{description}

\end{theorem}

\begin{proof} We prove {\bf (a)}, {\bf (e)} and {\bf (g)} only. The rest can be proved in a similar manner.

\begin{description}\item [(a)] Assume, to the first order of $e$, that
$\, \undersym{\lambda}{i}^{\mu} = \,\undersym{D}{i}^{\mu} + \epsilon \,\undersym{H}{i}^{\mu} + \varepsilon \,\undersym{K}{i}^{\mu}.$
Then, in view of
$\,\undersym{\lambda}{i}^{\mu}\,\,\undersym{\lambda}{i}_{\nu} = \delta^{\mu}_{\nu},$
we have
$\,\undersym{D}{i}^{\mu} = \,\undersym{\delta}{i}_{\mu},$ \
$\epsilon\{\,\undersym{H}{i}^{\mu}\,\,\undersym{\delta}{i}_{\nu} + \,\undersym{h}{i}_{\nu}\,\,\undersym{\delta}{i}_{\mu}\} = 0$ and
$\varepsilon\{\,\undersym{K}{i}^{\mu}\,\,\undersym{\delta}{i}_{\nu} + \,\undersym{k}{i}_{\nu}\,\,\undersym{\delta}{i}_{\mu}\} = 0.$
By the second relation above, we get $\,\undersym{H}{i}^{\mu}\,\,\undersym{\delta}{i}_{\nu} = - \,\,\,\undersym{h}{i}_{\nu}\,\,\undersym{\delta}{i}_{\mu}.$
Multiplying by $\undersym{\delta}{j}_{\nu}$, we obtain
$\label{txr}\,\undersym{H}{j}^{\mu} = - (\,\undersym{\delta}{j}_{\nu}\,\,\undersym{h}{i}_{\nu})
\,\,\undersym{\delta}{i}_{\mu} = - \,\,\undersym{h}{i}_{j}\,\,\undersym{\delta}{i}_{\mu} := - \,\undersym{h}{\mu}_{j}.$
Similarily, by $\varepsilon\{\,\undersym{K}{i}^{\mu}\,\,\undersym{\delta}{i}_{\nu} + \,\undersym{k}{i}_{\nu}\,\,\undersym{\delta}{i}_{\mu}\} = 0,$
we conclude that $\label{ttxr}\,\undersym{K}{j}^{\mu} = - \,\undersym{k}{\mu}_{j}.$ The expression of $\,\,\undersym{\lambda}{i}^{a}$ is
obtained in a similar manner, using the relation $\,\undersym{\lambda}{i}^{a}\,\undersym{\lambda}{i}_{b} = \delta^{a}_{b}.$

\vspace{- 0.05 cm}\item [(e)] By {\bf (a)}, we have
\begin{eqnarray*}C^{a}_{bc} &=& \,\undersym{\lambda}{i}^{a}\,\dot{\partial_{c}}\,\,\undersym{\lambda}{i}_{b} = \{\,\undersym{\delta}{i}_{a} -
\epsilon\,\undersym{h}{a}_{i} - \varepsilon\,\,\undersym{k}{a}_{i} + O(e^{2})\}\{\varepsilon \,\,\undersym{k}{i}_{b; c}\} =
\varepsilon \,\,\undersym{k}{a}_{b; c} + O(e^{2}).\\
\end{eqnarray*}

\vspace{- 0.8 cm}
\item [(g)] By {\bf (c)}, we have $g^{ad} = \delta_{ad} - \epsilon z_{ad}(x) - \varepsilon w_{ad}(y) + O(e^{2})$. Moreover,
$$g_{cd; b} = \varepsilon w_{cd; b}(y) + O(e^{2}), \,g_{bd; c} = \varepsilon w_{bd; c}(y) + O(e^{2}), \, g_{bc; d} = \varepsilon w_{bc; d}(y) + O(e^{2}).$$
Consequently,
\begin{eqnarray*}\overcirc{C}^{a}_{bc} &=& \frac{1}{2} \,g^{ad}\{g_{cd; b} + g_{bd; c} - g_{bc; d}\}\\
&\backsimeq& \frac{1}{2}\,\varepsilon \delta_{ad}\{w_{dc; b} + w_{db; c} - w_{bc; d}\}(y) =  \frac{1}{2}\,\varepsilon \{w_{ab; c} + w_{ac; b} - w_{bc; a}\}(y).\\
\end{eqnarray*}

\end{description}
\vspace{- 1.5 cm}\end{proof}

In the light of the above theorem, using the relevant definitions, we obtain

\begin{proposition}\label{Sumx} To the first order of $\varepsilon$, the following hold:
\begin{description}

\item [(a)] $\epsilon_{ab} \backsimeq \varepsilon(\,\undersym{k}{d}_{a; bd} - \,\undersym{k}{d}_{b; ad})(y),$

\item [(b)] $\theta_{ab}\backsimeq \varepsilon
(\,\undersym{k}{d}_{a; bd} + \,\undersym{k}{d}_{b; ad} - 2\,\undersym{k}{d}_{d; ab})(y),$

\item [(c)]$\xi_{ab}\backsimeq \varepsilon\{\,\undersym{k}{a}_{b; dd} - \frac{1}{2}\,(w_{ab; dd} +
w_{ad; bd} - w_{bd; ad})\}(y)$,

\item [(d)] $\psi_{ab} \backsimeq \varepsilon\{(\,\undersym{k}{d}_{a; bd} + \,\undersym{k}{d}_{b; ad}) - (w_{ad; bd} +
w_{bd; ad} - w_{ab; dd})\}(y)$,

\item [(e)]  $\chi_{ab}\backsimeq \varepsilon(\,\undersym{k}{d}_{a; bd} - \,\undersym{k}{d}_{b; ad})(y).$
\end{description}

\end{proposition}

Let $W_{abc}: = \frac{1}{2}(w_{ab; c} + w_{ac; b} - w_{bc; a}).$ Then the following tensors contain no terms of first order:

\begin{proposition} \label {uvx}To the second order of $\varepsilon$, the following hold:

\begin{description}

\item[(a)] $\gamma_{ab}\backsimeq
\varepsilon^{2}\{(\,\,\undersym{k}{c}_{d; c} - \,\,\undersym{k}{c}_{c; d})(\,\,\undersym{k}{a}_{b; d} -
W_{abd})\}(y)$

\item[(b)] $\phi_{ab}\backsimeq \varepsilon^{2}\{
(\,\,\undersym{k}{c}_{d; c} - \,\,\undersym{k}{c}_{c; d})(\,\,\undersym{k}{d}_{a; b} +
\,\,\undersym{k}{d}_{b; a} - 2W_{dab})\}(y).$

\item[(c)] $\eta_{ab}\backsimeq \varepsilon^{2}\{
(\,\,\undersym{k}{c}_{d; c} - \,\,\undersym{k}{c}_{c; d})(\,\,\undersym{k}{d}_{a; b} -
\,\,\undersym{h}{d}_{b; a})\}(y).$

\item [(d)] $\omega_{ab}\backsimeq \varepsilon^{2}\{
(\,\,\undersym{k}{d}_{a: c} - W_{dac})(\,\,\undersym{k}{c}_{b; d} - W_{cbd})\}(y).$

\item [(e)] $\sigma_{ab}\backsimeq  \varepsilon^{2}\{
(\,\,\undersym{k}{d}_{c; a} - W_{dca})(\,\,\undersym{k}{c}_{d; b} - W_{cdb})\}(y).$

\item [(f)] $\alpha_{ab}\backsimeq  \varepsilon^{2}\{(\,\,\undersym{k}{c}_{a; c} - \,\,\undersym{k}{c}_{c; a})
(\,\,\undersym{k}{c}_{b; c} - \,\,\undersym{k}{c}_{c; b})\}(y).$

\item [(g)] $h_{ab}\backsimeq \varepsilon^{2}\{(\,\,\undersym{k}{d}_{c; a} -
W_{dca})(\,\,\undersym{k}{c}_{b; d} - W_{cbd}) + (\,\,\undersym{k}{d}_{a; c} -
W_{dac})(\,\,\undersym{k}{c}_{d; b} - W_{cdb})\}(y).$

\item [(h)] $U_{ab}\backsimeq \varepsilon^{2}\{(\,\,\undersym{k}{d}_{c; a} -
W_{dca})(\,\,\undersym{k}{c}_{b; d} - W_{cbd}) - (\,\,\undersym{k}{d}_{a; c} -
W_{dac})(\,\,\undersym{k}{c}_{d; b} - W_{cdb})\}(y).$

\end{description}
\end{proposition}

\begin{corollary} Assume that (\ref{linear}) holds. Then, up to the first order of approximation, the vertical fundamental geometric objects of the
EAP-space are functions of the directional argument $y$ only.
\end{corollary}

\end{mysect}


\begin{mysect}{Linearization in the Cartan-type case}
We now consider the Cartan-type case. In this case, the nonlinear connection coefficients $N^{a}_{\mu}$ are given by
$N^{a}_{\mu} = y^{b}(\,\,\undersym{\lambda}{i}^{a}\partial_{\mu}\,\undersym{\lambda}{i}_{b})$\cite{EAP}. Consequently,
all geometric objects of the EAP-space are expressed explicitely in terms of the fundamental vector fields forming
the parallelization.

\bigskip

Similar to Theorem \ref{ATEF}, we have the following
\begin{theorem}\label{fc} Assume that the canonical $d$-connection is of Cartan type. Then, up to the first order
approximation, the following hold:
\begin{description}
\item [(a)] $N^{a}_{\mu} \backsimeq \epsilon\{y^{b}(\,\,\undersym{h}{a}_{b, \mu}(x))\}$.

\item [(b)]$\Gamma^{\alpha}_{\mu\nu} \backsimeq \epsilon(\,\undersym{h}{\alpha}_{\mu, \nu})(x), \ \ \
\Gamma^{a}_{b\nu} \backsimeq \epsilon(\,\undersym{h}{a}_{b, \nu})(x).$ Consequently, $N_{\nu}: = N^{a}_{\nu; a} \backsimeq \epsilon \,\,\undersym{h}{a}_{a, \nu}(x).$

\item [(c)] $\Lambda^{\alpha}_{\mu\nu} \backsimeq \epsilon(\,\undersym{h}{\alpha}_{\mu, \nu} - \,\undersym{h}{\alpha}_{\nu, \mu})(x).$

\item [(d)]$C_{\mu} \backsimeq \epsilon (\,\undersym{h}{\alpha}_{\mu, \alpha} - \,\undersym{h}{\alpha}_{\alpha, \mu})(x).$

\item [(e)] $\,\overcirc{\Gamma}^{\alpha}_{\mu\nu} \backsimeq \frac{1}{2}\,\epsilon(z_{\mu\alpha, \nu} + z_{\nu\alpha, \mu} - z_{\mu\nu, \alpha})(x).$

\item [(f)]$\gamma^{\alpha}_{\mu\nu} \backsimeq \epsilon\{\,\undersym{h}{\alpha}_{\mu, \nu} - \frac{1}{2}\,(z_{\mu\alpha, \nu} +
z_{\nu\alpha, \mu} - z_{\mu\nu, \alpha})\}(x).$

\item [(g)]$\Omega^{\alpha}_{\mu\nu} \backsimeq \epsilon
\{(\,\undersym{h}{\alpha}_{\mu,\nu} + \,\undersym{h}{\alpha}_{\nu, \mu}) - (z_{\mu\alpha, \nu} +
z_{\nu\alpha, \mu} - z_{\mu\nu, \alpha})\}(x).$
\end{description}

\begin{proof} We prove {\bf (a)} only. The rest is similar. We have
\begin{eqnarray*}N^{a}_{\mu} &=& y^{b}(\,\,\undersym{\lambda}{i}^{a}\partial_{\mu} \,\undersym{\lambda}{i}_{b})\\&=&y^{b}\{\,\undersym{\delta}{i}_{a} -
\epsilon\,\,\undersym{h}{a}_{i} - \varepsilon\,\,\undersym{k}{a}_{i} + O(e^{2})\}\{\epsilon \,\,\undersym{h}{i}_{b, \mu}\} =
\epsilon (y^{b}\,\,\undersym{h}{a}_{b, \mu}) + O(e^{2}).\\
\end{eqnarray*}

\vspace{ - 1 cm}\end{proof}
\end{theorem}
In the next two propositions, we assume that the canonical $d$-connection is of {\bf Cartan type}. Then, using the relevant definitions, taking into account
Theorem \ref{fc} and setting $Z_{\beta\mu\nu}: = \frac{1}{2}(z_{\beta\mu, \nu} + z_{\beta\nu, \mu} - z_{\nu\mu, \beta}),$ we obtain

\begin{proposition}\label{Summary} To the first order of $\epsilon$, we have

\begin{description}

\item [(a)] $\epsilon_{\mu\nu}\backsimeq \epsilon(\,\undersym{h}{\alpha}_{\mu, \alpha\nu} - \,\undersym{h}{\alpha}_{\nu, \alpha\mu})(x),$

\item [(b)] $\theta_{\mu\nu}\backsimeq \epsilon
(\,\undersym{h}{\alpha}_{\mu, \alpha\nu} + \,\undersym{h}{\alpha}_{\nu, \alpha\mu} - 2\,\undersym{h}{\alpha}_{\alpha, \mu\nu})(x),$

\item [(c)]$\xi_{\mu\nu} \backsimeq \epsilon\{\,\undersym{h}{\mu}_{\nu, \alpha\alpha} - Z_{\mu\nu\alpha, \alpha}\}(x)$,

\item [(d)] $\psi_{\mu\nu}\backsimeq \epsilon\{(\,\undersym{h}{\alpha}_{\mu,\nu\alpha} + \,\undersym{h}{\alpha}_{\nu, \mu\alpha}) -
2Z_{\alpha\mu\nu, \alpha}\}(x)$,

\item [(e)]  $\chi_{\mu\nu} \backsimeq \epsilon(\,\undersym{h}{\alpha}_{\mu, \nu\alpha} - \,\undersym{h}{\alpha}_{\nu, \mu\alpha})(x).$
\end{description}

\end{proposition}
The following tensors are of order $\epsilon^{2}$.
\begin{proposition}\label{uv} To the second order of $\epsilon$, we have
\begin{description}

\item[(a)] $\gamma_{\mu\nu}\backsimeq
\epsilon^{2}\{(\,\,\undersym{h}{\beta}_{\alpha, \beta} - \,\,\undersym{h}{\beta}_{\beta, \alpha})(\,\,\undersym{h}{\mu}_{\nu, \alpha} -
Z_{\mu\nu\alpha})\}(x)$

\item[(b)] $\phi_{\mu\nu}\backsimeq \epsilon^{2}\{
(\,\,\undersym{h}{\beta}_{\alpha, \beta} - \,\,\undersym{h}{\beta}_{\beta, \alpha})(\,\,\undersym{h}{\alpha}_{\mu, \nu} +
\,\,\undersym{h}{\alpha}_{\nu, \mu} - 2Z_{\alpha\mu\nu})\}(x).$

\item[(c)] $\eta_{\mu\nu}\backsimeq \epsilon^{2}\{
(\,\,\undersym{h}{\beta}_{\alpha, \beta} - \,\,\undersym{h}{\beta}_{\beta, \alpha})(\,\,\undersym{h}{\alpha}_{\mu, \nu} -
\,\,\undersym{h}{\alpha}_{\nu, \mu})\}(x).$

\item [(d)] $\omega_{\mu\nu}\backsimeq \epsilon^{2}\{
(\,\,\undersym{h}{\alpha}_{\mu, \beta} - Z_{\alpha\mu\beta})(\,\,\undersym{h}{\beta}_{\nu, \alpha} - Z_{\beta\nu\alpha})\}(x).$

\item [(e)] $\sigma_{\mu\nu}\backsimeq  \epsilon^{2}\{
(\,\,\undersym{h}{\alpha}_{\beta, \mu} - Z_{\alpha\beta\mu})(\,\,\undersym{h}{\beta}_{\alpha, \nu} - Z_{\beta\alpha\nu})\}(x).$

\item [(f)] $\alpha_{\mu\nu}\backsimeq  \epsilon^{2}\{(\,\,\undersym{h}{\beta}_{\mu, \beta} - \,\,\undersym{h}{\beta}_{\beta, \mu})
(\,\,\undersym{h}{\sigma}_{\nu, \sigma} - \,\,\undersym{h}{\sigma}_{\sigma, \nu})\}(x).$

\item [(g)] $h_{\mu\nu}\backsimeq \epsilon^{2}\{(\,\,\undersym{h}{\beta}_{\alpha, \mu} -
Z_{\beta\alpha\mu})(\,\,\undersym{h}{\alpha}_{\nu, \beta} - Z_{\alpha\nu\beta}) + (\,\,\undersym{h}{\beta}_{\mu, \alpha} -
Z_{\beta\mu\alpha})(\,\,\undersym{h}{\alpha}_{\beta, \nu} - Z_{\alpha\beta\nu})\}(x).$

\item [(h)] $U_{\mu\nu}\backsimeq \epsilon^{2}\{(\,\,\undersym{h}{\beta}_{\alpha, \mu} -
Z_{\beta\alpha\mu})(\,\,\undersym{h}{\alpha}_{\nu, \beta} - Z_{\alpha\nu\beta}) - (\,\,\undersym{h}{\beta}_{\mu, \alpha} -
Z_{\beta\mu\alpha})(\,\,\undersym{h}{\alpha}_{\beta, \nu} - Z_{\alpha\beta\nu})\}(x).$
\end{description}
\end{proposition}

The relation equivalent to (\ref{EEEXz})  in the Cartan type case is given by
\begin{equation}\begin{split}\label{zxzz}\\[- 0.4 cm]0& = E_{\mu\nu} := g_{\mu\nu}H - 2H_{\mu\nu} -
2C_{\mu}C_{\nu} - 2g_{\mu\nu}(C^{\epsilon}\!\,_{|\epsilon} - C^{\epsilon}C_{\epsilon})- 2C^{\epsilon}\Lambda_{\mu\epsilon\nu}
+ 2C_{\nu|\mu}\\& \ \ \ \ \ \ \ \ \ \ \ \ \ - \ 2g^{\epsilon\alpha}\Lambda_{\mu\nu\alpha|\epsilon}
 - 2N_{\epsilon}(\Lambda^{\epsilon}\,_{\nu\mu} -  \Lambda_{\mu\nu}\,^{\epsilon})
+ 2g_{\mu\nu}C^{\epsilon}N_{\epsilon} - 2C_{\mu}N_{\nu}.\end{split}\end{equation}

Consequently, $C_{\mu}N_{\nu}$, $C^{\mu}N_{\mu}$, $N^{\beta}\Lambda_{\mu\nu\beta}$ and $N^{\beta}\Lambda_{\beta\mu\nu}$ are the additional
terms appearing in the horizontal field equations under the Cartan condition that have no counterparts in the field
equations obtained in the context of the GFT.

\begin{proposition}\label{extra} The following holds:
\begin{description}

\item [(a)] $C_{\mu}N_{\nu} = \,\epsilon^{2}\{(\,\undersym{h}{\alpha}_{\mu, \alpha} -
\undersym{h}{\alpha}_{\alpha, \mu})\,\undersym{h}{c}_{c, \nu}\}(x).$
\item [(b)] $N^{\beta}\Lambda_{\mu\nu\beta} = \epsilon^{2}\{\,\undersym{h}{c}_{c, \beta}(\,\undersym
{h}{\mu}_{\nu, \beta} - (\,\undersym{h}{\mu}_{\beta, \nu})\}(x).$

\item [(c)] $N^{\beta}\Lambda_{\beta\mu\nu} =   \,\epsilon^{2}\{\,\undersym{h}{c}_{c, \beta}(\,\undersym{h}{\beta}_{\mu, \nu} -
\,\undersym{h}{\beta}_{\nu, \mu})\}(x).$

\end{description}
\end{proposition}

In view of Propositions (\ref{Summary}), (\ref{uv}) and (\ref{extra}), we have
\begin{corollary}\label{same} Assume that (\ref{linear}) holds. Then, up to the first order of approximation,
the purely horizontal fundamental geometric objects of the
EAP-space in the Cartan type case
are {identical} to their corresponding counterparts in the context of classical AP-space \cite{aa} and are
functions of the positional argument $x$ only.
This is because all extra terms appearing in the horizontal field equations in ETUFT are of order $\epsilon^{2}$.
Consequently, the horizontal field equations
under the Cartan condition {coincides} with the GFT up to the {first} order of approximation.
\end{corollary}

One reading of Corollary \ref{same} is that our constructed unified field theory seems to be a
plausable generalization of the GFT. Not only do the
horizontal field equations in the Berwald-type case coincide with those of the GFT \cite{WNA},
but also {\it the horizontal field equations in the {Cartan-type}
case {coincide} with the GFT in the first order of approximation}.
This also means that our theory (under the Cartan type condition) {\it differs} from the
GFT {\it only} when dealing with {\it strong} fields, that is, in the second (and higher orders) of $\epsilon$.
In other words, the {\it two theories actually coincide when dealing with weak fields}.

\end{mysect}

\begin{mysect}{Approximate solutions of the vertical field equations}

In this final section, we will examine the solutions of the vertical field equations corresponding to the
first order of approximation.
\bigskip

By (\ref{SYMx}) and (\ref{xems}), the vertical Einstein tensor is expressed in terms of the fundamental tensors in the form:

\begin{equation}\label{SYM}\ \overcirc{S}_{ab} - \frac{1}{2} \, h_{ab} \ \overcirc{\cal S} =
\frac{1}{2}\,g_{ab}(\bar{\sigma} - \bar{h})  - (\sigma_{ab} - h_{ab}).\end{equation}

As easily checked, in view of Theorem \ref{ATEF} {\bf (c)} and Proposition \ref{uvx} {\bf (e)} and {\bf (g)},
(\ref{SYM}) reduces in the first
order of approximation to
\begin{equation}\label{fo} (\,\overcirc{S}_{ab})^{(1)} - \frac{1}{2}\,\delta_{ab}\,\,\overcirc{\cal S}^{(1)} = 0.\\[- 0.2 cm]\end{equation}
Contracting, we get $\,\overcirc{\cal S}^{(1)} = 0$, so that (\ref{fo}) reduces to
\begin{equation}\label{zerox}(\,\overcirc{S}_{ab})^{(1)} = 0.\end{equation}
\vspace{- 0.1 cm}Now, we have \cite{WNA} 
\begin{equation}\label{Sx}\overcirc{S}_{ab} = - \frac{1}{2}\,(\theta_{ab} - \psi_{ab} + \phi_{ab}) +
\omega_{ab}.\end{equation}
Consequently, in view of Proposition \ref{uvx} ($(\phi_{ab})^{(1)} = (\omega_{ab})^{(1)} = 0$), (\ref{Sx}) reduces in the first order of approximation to
\begin{equation}(\,\overcirc{S}_{ab})^{(1)} = \frac{1}{2}\{(\psi_{ab})^{(1)} - (\theta_{ab})^{(1)}\},\end{equation}
so that, by Proposition \ref{Sumx} {\bf (b)} and {\bf (d)}, noting that $w_{dd} = 2\,\undersym{k}{d}_{d}$, we obtain
\begin{eqnarray*}(\,\overcirc{S}_{ab})^{(1)} &=& \frac{1}{2}\,\{(\,\undersym{k}{d}_{a; bd} + \,\undersym{k}{d}_{b; ad})
- w_{da; bd} - w_{db; ad} + w_{ab; dd}\\
&& - \ (\ \,\undersym{k}{d}_{a; db} +
\,\undersym{k}{d}_{b; ad}) + \ 2\,\,\undersym{k}{d}_{d; ab}\}\\&=&\frac{1}{2}\,\{w_{ab; dd} -
w_{ad; bd} - w_{bd; ad} + w_{dd; ab}\}.\\
\\[- 1cm]\end{eqnarray*}
The above equation, together with (\ref{zerox}), gives
\begin{equation}\label{bp}w_{ab; dd} - w_{ad; bd} - w_{bd; ad} - w_{dd; ab} = 0.\end{equation}
\\[- 0.5 cm] Consequently,
\begin{eqnarray*}w_{ab; dd} &=& w_{ad; bd} + w_{bd; ad} -
w_{dd; ab}\\&=&(w_{ad; bd} - \frac{1}{2}\, w_{dd; ab}) + (w_{bd; ad} - \frac{1}{2} \,
w_{dd; ab}),\\\end{eqnarray*}
\\[ - 1 cm] which can be expressed in the form
\begin{equation} w_{ab; dd} = \big\{(w_{ad; d} - \frac{1}{2}\, w_{dd; a})_{; b} +
(w_{bd; d} - \frac{1}{2} \, w_{dd; b})_{; a}\big\}.\end{equation}
Hence, if $(\,\, \overcirc{C}^{a}_{dd})^{(1)} = 0$, that is, $w_{ad; d} = \frac{1}{2}\, w_{dd; a}$,
then the symmetric part of the vertical field equations in this case reduces to
\begin{equation}\label{wave}w_{ab; dd} = \big\{\frac{\partial^{2} }{\partial (y^{1})^{2}} + \frac{\partial^{2} }{\partial (y^{2})^{2}} +
\frac{\partial^{2} }{\partial (y^{3})^{2}} + \frac{\partial^{2} }{\partial (y^{4})^{2}}\big\}\, w_{ab} = 0.\end{equation}
In view of the interpretation of the variable $y^{a}$, (\ref{wave}) represents,  {\it \`a la Kaluza}, a {\it wave equation
in the  micro-internal dimension}.

\bigskip

On the other hand, by (\ref{MAXTx}) and (\ref{po}), the {vertical} generalized {electromagnetic field} is given by
\begin{equation}\label{boz} F_{ab} :=  (\gamma_{ab} - \xi_{ab} + \eta_{ab}),\end{equation}
\vspace{- 0.6 cm}\begin{equation}\label{poz}F_{ab} = \dot{\partial_b}C_{a} - \dot{\partial_a}C_{b}.\end{equation}

\noindent By Proposition \ref{uvx}, (\ref{boz}) and (\ref{poz}) reduce in the first approximation to
\begin{equation}\label{SSX}
(C_{a; b})^{(1)} - (C_{b; a})^{(1)} = - (\xi_{ba})^{(1)}\end{equation}
Now,  in the first order of $\varepsilon$,
\begin{equation}C_{a} \backsimeq \varepsilon(\,\undersym{k}{d}_{a; d} - \,\undersym{k}{d}_{d; a}) \backsimeq \varepsilon (C_{a})^{(1)}; \ \ \ \
C_{a; b} \backsimeq \varepsilon(\,\undersym{k}{d}_{a; db} - \,\undersym{k}{d}_{d; ab}) \backsimeq\varepsilon (C_{a; b})^{(1)}.\end{equation}
Consequently,
\begin{equation}\label{vq}(C_{a; b})^{(1)} = (C_{a})^{(1)}\!\,_{;b} = \,\undersym{k}{d}_{a; db} - \,\undersym{k}{d}_{d; ab}.\end{equation}
Contraction of relation (\ref{bp}) yields
\begin{equation}w_{dd; aa} = w_{ad; ad},\end{equation}
equivalently,
\begin{equation}\label{rr}\,\undersym{k}{d}_{d; aa} = \,\,\undersym{k}{d}_{a; da}.\end{equation}

\noindent Hence, in view of (\ref{vq}) and (\ref{rr}), we conclude that
\begin{equation}\label{pp}(C_{d; d})^{(1)} = (C_{d})^{(1)}\!\,_{; d} = 0\end{equation}

\noindent Differentiating (\ref{SSX}), taking into account (\ref{pp}), 
we obtain
\begin{equation}(C_{a; dd})^{(1)} = (C_{a})^{(1)}\!\,_{; dd} = - (\xi_{da})^{(1)}\!\,_{; d} = - (\xi_{da; d})^{(1)}\end{equation}
Consequently,  by $J^{a} := F^{ab}\!\,_{{o\atop{|}}b}$, where $J_{a}$ is the vertical current density, we obtain

\begin{equation}\big\{\frac{\partial^{2} }{\partial (y^{1})^{2}} + \frac{\partial^{2} }{\partial (y^{2})^{2}} +
\frac{\partial^{2} }{\partial (y^{3})^{2}} + \frac{\partial^{2} }{\partial (y^{4})^{2}}\big\}\, (C_{a})^{(1)} = - (J_{a})^{(1)}.\end{equation}

\noindent Hence, if $(J_{a})^{({1})} = 0$, then $(C_{a})^{(1)}$ again satisfies a {\it wave equation in the micro-internal dimension}.

\bigskip

The vertical part of the field equations gives rise to two different $4$-dimensional
Laplace equations which would be wave equations if the metric was not {positive definite}, but had {Lorentz signature}. The meaning of
this result, however, will be clarified if a clear {\it physical interpretation of the vector $y^{a}$ is given}.

\end{mysect}


\bigskip


We end the paper with the following remarks and comments:

\begin{itemize}


\item As previously mentioned, a possible interpretation of the coordinates would to take $x_1$, $x_2$ and $x_3$ as {space}
coordinates and $x^{4}$ as {time} coordinates. On the other hand, the vector $y^{a}$ is attached as an {internal} variable to each $x^{\mu}$.
According to Miron\, \cite{V}, $y^{a}$ may be regarded as spacetime fluctuations (micro internal freedom) associated to the point
$x^{\mu}$. It could be argued that this interpretation, however, seems somewhat incompatible with (\ref{linear}). This is because we are
actually combining dimensions of {different} natures, namely, a {\bf macro} dimension $x$ and a {\bf micro} dimension $y$. To avoid this
(apparent) incompatibility, we shall take $x^{\mu}$ as stated above and {\it leave the meaning of the directional argument $y^{a}$ unspecified, at least for the present.}

\item In any geometric field theory, authors try to attribute physical meaning to the geometric objects
present in the theory. One way to accomplish
this is to compare the new theory with previously existing field theories. The GFT is a
field theory unifying electromagnetism and gravitation. It
is thus a theory that generalizes both Maxwell's electromagnetic theory and Einstein's general theory of relativity.
The comparison of the GFT with these two theories,
resulted in attributing some physical meaning to certain geometric objects occuring in the GFT (using
certain systems of units). A new scheme,
called {\bf Type Analysis} \cite{W}, has also been suggested in the context of the GFT.
Its aim is to test the AP-geometry (the underlying geometry of the GFT) for representing certain physical fields.
The procedure of Type Analysis is usually applied before solving the field equations.
It is a {covariant} procedure formulated in the language of tensors. Some tensors are used to characterize
the AP-space under consideration. Roughly, the types of spaces considered are as follows:
spaces with or without electromagnetic field, spaces with or without gravitational fields, and, finally,
spaces in which both fields are present. The strengths of the fields involved are also taken into account.
More precisely, some tensors are defined indicating the presence or absence of
electomagnetic and gravitational fields. Not all possible combinations, however, are allowed.
There are certain constraints. For example, a space
with a non-vanishing electromagnetic field necessarily contains a non-vanishing gravitational
field.\footnote{This is easy to see. An electromagnetic field carries energy, that is
mass (since energy and mass are equivalent). But mass is the source of gravitational field.}

\item One of our future
aims is to generalize the procedure of Type Analysis to our constructed field
theory. It should be noted that this is not an easy task.
This is because both the physical and the mathematical scope of the ETUFT are
much wider and richer in content than the GFT.
Therefore, in the context of the ETUFT, an \lq\lq extension\rq\rq\, of the procedure of Type Analysis
(applied to the (horizontal) field equations in the Cartan type case and the vertical field equations
in the general case) is expected to be more complicated and much more involved.
Some of the reasons for such complications are the following:
First, the (horizontal) field equations in the Cartan type are more complex than their corresponding
counterparts in the GFT (they coincide only in the first order of approximation).
Secondly, we have two types of \lq\lq e\rq\rq \,
in our linearized field theory. Last, but not least, our geometric objects, unlike the classical AP-geometry
(in which the GFT is formulated), depend,
in general, on both the directional argument $x$ and the positional argument $y$.
\item On the other hand, an application of (some method similar to) the procedure of Type analysis
to the ETUFT may
help us illuminate the {physical role} of the nonlinear connection and give a plausible physical
interpretation to the directional argument $y$. This, in turn, may shed more light on the
possibility that the ETUFT - in addition to
unifying gravity and electromagnetism - could also describe some micro physical phenomena.

\newpage

\item In conclusion, we note that both Einstein's and Maxwell's equations in the ETUFT are {\bf doubled}.
This is due to the splitting of the field equations induced by the nonlinear connection.
An {\bf interpretation} of the new equations is needed. Perhaps it is possible to connect one of
Maxwell's fields with $SU(2)$ {gauge fields}. Such a step seems necessary, given the
spectacular success of Weinberg-Salam theory. This point will be the subject of future research.

\end{itemize}


\bibliographystyle{plain}

\end{document}